\documentclass[11pt,a4paper]{article}

\usepackage{amsmath}
\usepackage{amsfonts}
\usepackage{amsthm}
\usepackage[pdftex]{color,graphicx}
\usepackage{graphicx}
\usepackage{color}
\usepackage[colorlinks,linkcolor=blue,citecolor=red,urlcolor=black]{hyperref}

\usepackage{caption}
\usepackage{subcaption}

\newtheorem{Theorem}{Theorem}

\newtheorem{Proposition}{Proposition}
\newtheorem{Lemma}{Lemma}
\newtheorem{Remark}{Remark}
\theoremstyle{definition}

\newcommand{\dt}{{\mathrm{d}}t}

\newcommand{\dx}{{\mathrm{d}}x}
\newcommand{\dy}{{\mathrm{d}}y}
\newcommand{\Id}{{\mathbf{1}}}
\newcommand{\dom}{{\mathrm{dom}~}}

\usepackage{enumitem} % enumerar com letras
\usepackage{adjustbox}

\begin{document}

	\title{Spectrum of the Laplacian in waveguide shaped surfaces}

\author{Diana C. S. Bello \footnote{dianasuarez@estudante.ufscar.br} \\
{\small Departamento de Matemática – UFSCar, São Carlos, SP, 13565-905, Brazil}}

	\date{\today}
	
	\maketitle 
	
	\begin{abstract}

	\end{abstract}

Let $-\Delta_{\cal S}$ be the Laplace operator in ${\cal S} \subset \mathbb{R}^3$ on a waveguide shaped surfaces, i.e., ${\cal S}$ is built by translating a closed curve in a constant direction along an
unbounded spatial curve. Under the condition that the tangent vector
of the reference curve admits a finite limit at infinity, we find the essential spectrum of $-\Delta_{\cal S}$ and discuss conditions under which discrete eigenvalues emerge. Furthermore, we analyze the Laplacian in the case of a broken sheared waveguide shaped surface.
	
	\
	
	
	\
	
	
	%\noindent {Running head:} XXXXXXXXXXXXXX
	
	%

	\section{Introduction}

The spectral properties of the Laplace operator on unbounded domains 
$\Omega \subset \mathbb{R}^n$ have received considerable attention in recent years. In these domains, the existence of discrete eigenvalues is a non-trivial problem, influenced by both the geometric characteristics of the domain and the type of boundary conditions imposed, typically Dirichlet, Neumann, or mixed Dirichlet–Neumann conditions \cite{avishai, belloverri, bori2, bori1, briet, bulla, carini, duclosfk, dittrichakriz2, duclos, exnerseba, solomyak0, gold, daviddn, davidkriz, davidlu, davidmainpaper, olendski, renger, verri}. In these works the authors investigated the spectrum of the Laplace operator inside unbounded domains, such as strips or waveguides. Nevertheless, a natural question arises when one considers an alternative perspective: what spectral properties emerge if the waveguide surface itself is taken as the domain of analysis, without imposing conventional boundary conditions? In this scenario, the spectral behaviour is no longer governed primarily by external constraints, but rather by the intrinsic geometry of the surface. Understanding how this intrinsic geometry influences the spectral properties of the associated Laplace operator opens new avenues in spectral geometry and mathematical physics.

The main motivation of the present work is to explore  the spectrum of the Laplace operator  on waveguide-shaped surfaces without imposing boundary conditions. Specifically, inspired by the results in \cite{ohv, verri}, we demonstrate that the essential spectrum of the operator is determined by the asymptotic behaviour of the surface at infinity. Moreover, we establish sufficient conditions under which discrete spectrum arises.

Let $\{e_1, e_2, e_3\}$ the canonical basis of $\mathbb{R}$., and  $\xi: [0,1] \longrightarrow \mathbb{R}^2$ be a simple closed curve $C^2$ in the $\{e_2, e_3\}$ plane, parametrized by arc length $t$ (i.e., $\|\xi'(t)\|=1$, $\forall t \in [0,1]$) with $\xi(t)=  (\xi_1(t), \xi_2(t))$. Suppose that $f,g:\mathbb{R} \longrightarrow \mathbb{R}$ are locally Lipschitz continuous functions, differentiable almost everywhere, with derivatives $f', g' \in L^\infty(\mathbb{R})$. Next, we consider the spatial curve 
\begin{equation*}
	r(x)= (x, f(x), g(x)), \quad x \in \mathbb{R}.
	\end{equation*}
Define the mapping
\begin{equation*}
		\begin{array}{lcll}
			\mathcal{P} : & \mathbb{R} \times {\cal C}  & \longrightarrow & \mathbb{R}^3\\
			& (x,t)          & \longmapsto     & r(x) + \xi_1(t) e_2 +\xi_2(t) e_3 ,
		\end{array}
	\end{equation*}
where  ${\cal C}$ denotes the circle with length equal to one,  and the surface given by
\begin{equation*}
		{\cal S} := \mathcal{P} (\mathbb{R} \times {\cal C}).
	\end{equation*} 
such that, ${\cal S}$ is obtained by translating the curve $\xi$ along $r(x)$, as illustrated in Figure \ref{fig_surf}.  

	Let  $-\Delta_{\cal S}$ the Laplace operator in ${\cal S}$, i.e., , the self-adjoint operator associated with the quadratic form	%
	\begin{equation*}
Q_{\cal S} (\psi) 
		= \int_{\cal S} |\nabla \psi|^2 {\rm d}{\bf x}, \quad \dom Q_{\cal S}
		= {\cal H}^1({\cal S}),
	\end{equation*}
	where ${\bf x} = (x, y,z)$ denotes a point of ${\cal S}$ and $\nabla \psi$ denotes the gradient vector of  $\psi$. The objective of this paper is to study the spectral problem of $-\Delta_{\cal S}$. Motivated by \cite{verri}, we assume that  
\begin{equation}\label{eq:condlim_surf}
    \lim_{|x| \to \infty} f'(x) =: \beta_1, \quad  \lim_{|x| \to \infty} g'(x) =: \beta_2, \quad \beta_1, \beta_2 \in \mathbb{R}.
\end{equation}

 \begin{figure}[ht]
    \centering
    \hspace{0.5cm}
    \begin{tabular}{cc}
    \adjustbox{valign=m}{\subfloat[Surface ${\cal S}$.\label{subfig-1:dummy}]{%
          \includegraphics[width=0.4\textwidth]{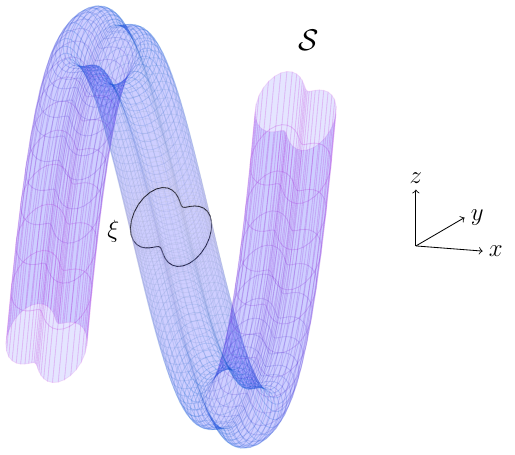}}}
    &      
    \adjustbox{valign=m}{\begin{tabular}{@{}c@{}}
    \subfloat[$xy$-plane view. \label{subfig-2:dummy}]{%
          \hspace{1cm}\includegraphics[width=0.07\textwidth]{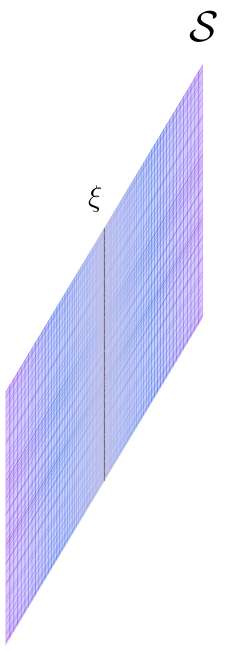} 
          \hspace{1cm}}  
          \vspace{0.5cm}
          \\ 
    \subfloat[$xz$-plane view.\label{subfig-3:dummy}]{%
          \hspace{1cm}\includegraphics[width=0.07\textwidth]{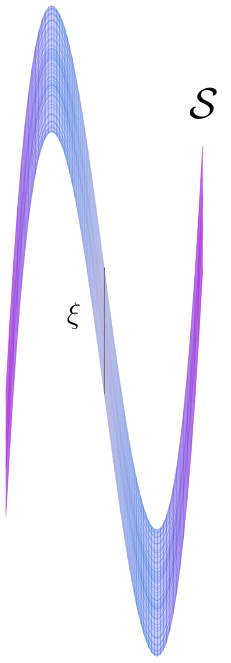} \hspace{1cm}} 
    \end{tabular}}
    \end{tabular}
    \caption{Waveguide shaped surface.}\label{fig_surf}
  \end{figure}

Our first objective is to characterize the essential spectrum of $-\Delta_{\cal S}$. For this purpose, we consider the two-dimensional operator
\begin{equation}\label{fourpp:surf_surf}
    		T_{\beta_1, \beta_2}:=- \frac{\partial}{\partial t} \left[ \left(\frac{1+\beta_1^2 + \beta_2^2}{h_{\beta_1,\beta_2}^2}\right) \partial_t \right] + (1+\beta_1^2 + \beta_2^2) \left[\left(\frac{(h_{\beta_1,\beta_2}^2)'}{4 h_{\beta_1,\beta_2}^4} \right)' + \left( \frac{(h_{\beta_1,\beta_2}^2)'}{4 h_{\beta_1,\beta_2}^3}\right)^2\right] ,
	\end{equation}
 \begin{equation*}
     \dom T_{\beta_1,\beta_2} :=\{v \in {\cal H}^1({\cal C}):  T_{\beta_1,\beta_2} v \in L^2({\cal C}), \; v'(0)=v'(1)\},
 \end{equation*}
 where $h_{\beta_1,\beta_2}^2 := 1+(\beta_1 \xi'_2 - \beta_2 \xi'_1)^2$. Denote by $E_1(0)$ the first eigenvalue of $T_{\beta_1,\beta_2}$ and by $\chi$ its corresponding normalized eigenfunction. Since $T_{\beta_1,\beta_2}$ is an elliptic operator with real coefficients, $E_1(0)$ is simple.  
 
\vspace{0.2cm}	
\begin{Theorem}\label{theo:surfess_surf}
	Suppose that the conditions in \eqref{eq:condlim_surf} hold. Then, $$\sigma_{ess}(-\Delta_{\cal S}) = [E_1(0), \infty).$$
	\end{Theorem}

The proof of Theorem \ref{theo:surfess_surf} is presented in Section \ref{sec3}.

 The next step is to analyze the existence of discrete eigenvalues for  $-\Delta_{\cal S}$. Define, for all $x \in \mathbb{R}$,  the functions  
 \begin{equation*}
   A(x) := \int_{\cal C} \frac{1}{h^2} |\chi^2| \dt, \quad
   B(x):= \int_{\cal C}    \left[(f' \xi'_1+ g' \xi'_2) \frac{ \partial_t h^2}{2 h^4} +\frac{\partial}{\partial t}\left(\frac{f' \xi'_1+ g' \xi'_2}{h^2}\right)  -  \frac{\partial_x h^2}{2h^4}
   \right] |\chi|^2 \dt,
 \end{equation*}
 \begin{equation*}
     C(x):= \int_{\cal C} \left[\left (\frac{\partial_x h^2}{4h^3} \right)^2   - \frac{\partial}{\partial t}  \left( (f' \xi'_1+ g' \xi'_2) \frac{ \partial_x h^2}{4 h^4} \right)  -  (f' \xi'_1+ g' \xi'_2)\frac{ \partial_x h^2 \partial_t h^2}{8 h^6} \right] |\chi|^2 \dt,
 \end{equation*}
 \begin{equation*}
     D(x):= \int_{\cal C} \left(\frac{1}{h^2} |\chi'|^2 + \left[
  \frac{\partial}{\partial t} \left(\frac{\partial_t h^2}{4 h^4} \right) + 
 \left(\frac{ \partial_t h^2}{4 h^3}\right)^2 \right] |\chi|^2 \right)  \dt,
 \end{equation*}
 the constant 
 \begin{equation*}
     E:= \int_{\cal C} \left( \frac{1}{h_{\beta_1, \beta_2}^2} |\chi'|^2 + \left[ \left(\frac{(h_{\beta_1, \beta_2}^2)'}{4 h_{\beta_1, \beta_2}^4} \right)' +  \left(\frac{ (h_{\beta_1, \beta_2}^2)'}{4 h_{\beta_1, \beta_2}^3}\right)^2 \right] |\chi|^2 \right) \dt,
 \end{equation*}
and the function 
\begin{equation*}
  V(x):=  C(x) + (1+(f')^2 + (g')^2) D(x) -  
 (1+\beta_1^2 + \beta_2^2)  E, \quad x \in \mathbb{R}.  
\end{equation*}

\begin{Theorem}\label{exidisspc_surf}
	Suppose that the conditions in \eqref{eq:condlim_surf} hold, $V(x) \in L^1(\mathbb{R})$ and $\int_\mathbb{R} V(x)<0$. Then,  $$\inf \sigma (-\Delta_{\cal S}) < E_1(0),$$
 i.e., 
	$\sigma_{dis}(-\Delta_{\cal S}) \neq \emptyset$. 
\end{Theorem}

The proof of this result is given in Section \ref{sec4}. Note that the condition $\int_{\mathbb{R}} V(x) \dx<0$ implies the existence of discrete eigenvalues for $-\Delta_{\cal S}$. However, it is not a necessary condition for this happen. For instance,

\begin{Theorem}\label{exidisspc_surf2}
    Suppose that the conditions in  \eqref{eq:condlim_surf} hold, $V(x), B(x) \in L^1(\mathbb{R})$,  $\int_\mathbb{R} V(x)=0$ and $B(x)$ is not constant. Then,  $$\inf \sigma (-\Delta_{\cal S}) < E_1(0),$$
 i.e., 
	$\sigma_{dis}(-\Delta_{\cal S}) \neq \emptyset$. 
\end{Theorem}
The proof of this result is also presented in Section \ref{sec4}.

Now, motivate by \cite{belloverri}, we consider a ``broken sheared waveguide shaped surfaces'', in which $f(x)$ is a null function and $g(x)=\beta|x|$, for all $x \in \mathbb{R}$. Observe that in this case, we are not under the conditions of the Theorems \ref{exidisspc_surf} and \ref{exidisspc_surf2}. Thus, take $\beta \in (0, \infty)$, consider the spatial curve
\begin{equation*}
		r_\beta(x)= (x, 0, g(x)), \quad x \in \mathbb{R}.
	\end{equation*}
 Define the mapping
\begin{equation}\label{lmas_surfcant}
		\begin{array}{lcll}
			\mathcal{P}_\beta : & \mathbb{R} \times {\cal C}  & \longrightarrow & \mathbb{R}^3\\
			& (x,t)          & \longmapsto     & (x, \xi_1(t), \xi_2(t) + \beta |x|),
		\end{array}
	\end{equation}
and the surface  
\begin{equation*}
		{\cal S}_\beta := \mathcal{L}_\beta (\mathbb{R} \times {\cal C}).
	\end{equation*} 
Similarly to the surface $\cal S$, geometrically, ${\cal S}_\beta$ is obtained by translating the closed curve $\xi$ along the curve $r_\beta(x)$. Furthermore,  note that ${\cal S}_\beta$ is symmetric with respect to the $\{e_2,e_3\}$ plane, and it has a ``corner'' lying in this plane, as shown in the Figure \ref{fig:surf_beta}. 

\begin{figure}[ht]
	\centering
	\includegraphics[width=0.7\textwidth]{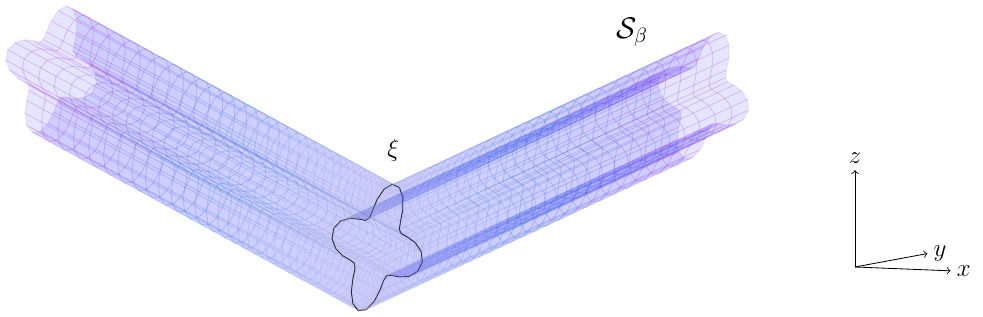}
	%\caption{Subfigure}
	\caption{Broken sheared waveguide shaped surfaces, with $\beta=0.5$.}
	\label{fig:surf_beta}
\end{figure}

	Denote by  $-\Delta_{{\cal S}_\beta}$ the Laplace operator in ${\cal S}_\beta$, i.e., the self-adjoint operator associated with the quadratic form 
	\begin{equation*}
Q_{{\cal S}_\beta} (\psi) 
		= \int_{{\cal S}_\beta} |\nabla \psi|^2 {\rm d}{\bf x}, \quad \dom Q_{{\cal S}_\beta}
		= {\cal H}^1({\cal S}_\beta).
	\end{equation*}

In this case, consider the two-dimensional operator  
\begin{equation*}
    		T(\beta):=- \frac{\partial}{\partial t} \left[ \left(\frac{1+\beta^2}{h_\beta^2}\right) \partial_t \psi\right] + (1+\beta^2) \left[  \left(\frac{(h_\beta^2)'}{4 h_\beta^4} \right)'  + \left| \frac{(h_\beta^2)'}{4 h_\beta^3}\right|^2  \right],
	\end{equation*}
 \begin{equation*}
     \dom T(\beta)=\{v \in {\cal H}^1({\cal C}):  T(\beta) v \in L^2({\cal C}), \; v'(0)=v'(1)\}.
 \end{equation*}
where $h_\beta^2(t) = 1+ \beta^2 (\xi'_1(t))^2$. Denote $E_1(\beta)$ and $\chi_1$ to be the first eigenvalue and the corresponding normalized eigenfunction associated to $T(\beta)$, respectively.

\begin{Theorem}\label{propress_surfcant}
	For each $\beta \in (0, \infty)$, the essential spectrum of  $-\Delta_{{\cal S}_\beta}$ is the interval $[E_1 (\beta),\infty)$. 
\end{Theorem}

In order to analyze the existence of discrete eigenvalues for $-\Delta_{{\cal S}_\beta}$. Define the constants 
 \begin{align*}
   A &:= \int_{\cal C} \frac{|\chi_1|^2}{h_\beta^2} \dt \quad \mbox{and} \quad 
   B:= \int_{\cal C}   \left( \left(\frac{\xi'_2(t)}{h_\beta^2} \right)' + \frac{ \xi'_2(t) (h_\beta^2)'}{2 h_\beta^4} \right) |\chi_1|^2 \dt.
 \end{align*}

\begin{Theorem}\label{exidisspc_surfcant}
    Suppose that $B$ a nonzero constant. For each  $\beta \in (0, \infty)$, one has  $\inf \sigma (-\Delta_{{\cal S}_\beta}) < E_1(\beta)$, i.e., 
	$\sigma_{dis}(-\Delta_{{\cal S}_\beta}) \neq \emptyset$.
\end{Theorem}

The proof of Theorems  \ref{propress_surfcant} and \ref{exidisspc_surfcant} are presented in Section \ref{sec5}.

\begin{Remark}\label{remidn1}{\rm 
Let $f(x)$ be a null function and $g(x)$ an even function, so that $\lim_{x \to \infty} g'(x)= \beta$. In this scenario, the waveguide shaped surface $\cal S$ would possess symmetry relative to the $\{e_2, e_3\}$ plane. Thus, following the ideas from the proofs in the Section \ref{sec5}, we can show the existence of elements in the discrete spectrum of $-\Delta_{\cal S}$.}
	\end{Remark}
%Consider $f(x)$ as a null function and $g(x)$ as even function satisfying $\lim_{x \to \infty} g'(x)= \beta$. Then the waveguide shaped surface ${\cal S}$ is  symmetric with respect to the $\{e_2, e_3\}$ plane. Following the ideas from the proofs in Section \ref{sec5}, we establish the existence of discrete spectrum elements for $-\Delta_{\cal S}$.

\begin{Remark}\label{remidn2}{\rm 
   The model considered in this work was inspired by \cite{ohv} and \cite{verri}. In  \cite{ohv} the authors study $3$-dimensional waveguide surfaces defined via curves in $\mathbb{R}^3$ possessing a Frenet frame, on the other hand, in \cite{verri}, the author considers a waveguide constructed through parallel transport of a cross-section along a space curve.}
	\end{Remark}

The present work is organized as follows. In Section \ref{sec2} we describe the usual changes of coordinates to work in the Hilbert space $L^2(\Sigma)$ with the usual metric. In Section \ref{sec3} we study the essential spectrum of  $-\Delta_{\cal S}$ as well as the proof of Theorem \ref{theo:surfess_surf}. Next, Section \ref{sec4} is dedicated to study the discrete spectrum of $-\Delta_{\cal S}$. Where, we present the proofs of Theorems  \ref{exidisspc_surf} and  \ref{exidisspc_surf2}. Finally, Section \ref{sec5} concerns the spectral analysis of the Laplace operator in a broken sheared waveguide shaped surfaces, in turns, containing the proofs of Theorems \ref{propress_surfcant} and \ref{exidisspc_surfcant}.

\section{Change of coordinates} \label{sec2}
		
Recall that  ${\cal S} = {\cal P}(\Sigma)$, where $\Sigma := \mathbb{R \times {\cal C}}$. Then, in this section we perform a change of coordinates such that $Q_{\cal S}(\psi)$ starts to act in the Hilbert space $L^2(\Sigma)$ instead of $L^2(\cal S)$.

Let $G=(G_{ij})$ be the metric induced by   ${\cal P}$, i.e.,
	\[G_{ij} = \langle {\cal G}_i, {\cal G}_j \rangle = G_{ji}, \quad i,j=1,2,\]
	where
	\[{\cal G}_1 = \frac{\partial {\cal P}}{\partial x}= (1, f'(x), g'(x)), \quad 
	{\cal G}_2 = \frac{\partial {\cal P}}{\partial t}=(0, \xi'_1(t), \xi'_2(t)).\]
	More precisely,
	\[G = \nabla {\cal P} \cdot (\nabla {\cal P})^t = \left(
	\begin{array}{cc}
		1 + (f'(x))^2+(g'(x))^2  & f'(x)\xi'_1(t) +  g'(x)\xi'_2(t)   \\
	f'(x)\xi'_1(t) +  g'(x)\xi'_2(t)    &  1
	\end{array} \right),\]
and $\det G = 1+ (f'(x)\xi'_2(t)-g'(x)\xi'_1(t))^2 \neq 0$, for all $(x,t) \in \Sigma$. Since ${\cal P}$ is a global diffeomorphism between $\Sigma$ and ${\cal S}$, a change of variables can be performed. 

Denote by $h(x,t) :=  h_{f',g'}(x,t) =\sqrt{\det G}$. The norm in the Hilbert space  $L^2(\Sigma, h(x,t) \dx \dt)$ is given by
\begin{equation*}
    \|\psi\|^2 := \int_{\Sigma} |\psi|^2 h \dx \dt.  
\end{equation*}

Now, we consider the unitary operator  
\begin{equation*}
		\begin{array}{llll}
			{\cal U}: &   L^2({\cal S})  &  \to &  L^2(\Sigma, h \,\dx \dt) \\
			&    \psi   &  \mapsto  &        \psi \circ {\cal P}
		\end{array},
	\end{equation*}
 and, we define
	\begin{align*}
		Q_{f',g'}(\psi)  := & \; Q_{\cal S}({\cal U}^{-1} \psi)  
		= \int_{\Sigma} \langle \nabla \psi, G^{-1} \nabla \psi \rangle \sqrt{{\rm det}\, G} \, \dx \dt  \nonumber \\ 
		= &  \int_{\Sigma}  \text{\small $\left(\left|\psi' - (f'(x) \xi'_1(t) + g'(x) \xi'_2(t)) \frac{\partial \psi}{\partial t} \right|^2 + (1 + (f'(x) \xi'_2(t) - g'(x) \xi'_1(t))^2) \left| \frac{\partial \psi}{\partial t} \right|^2 \right) \frac{\dx \dt}{h (x,t)}$} \nonumber \\
  	= &  \int_{\Sigma} \frac{\left|\psi' - (f'(x) \xi'_1(t) + g'(x) \xi'_2(t)) \partial_t \psi \right|^2}{h (x,t)} \dx \dt  + \int_{\Sigma} 
 h(x,t) \left| \partial_t \psi \right|^2 \dx \dt,
	\end{align*}
	\[\dom Q_{f',g'} := {\cal U}(\dom Q_{\cal S}),\]
	where
	$\psi' := \partial \psi / \partial x$. Since $f', g' \in L^\infty(\mathbb{R})$, one get $ \dom Q_{f',g'} = {\cal H}^1(\Sigma)$.  Denote by $H_{f',g'}$ the self-adjoint operator associated with the quadratic form
 $Q_{f',g'}(\psi)$. %Assim, podemos identificar $-\Delta_{\cal S}$ com o operador autoadjunto $H_{f',g'}$.

Considering an additional unitary operator 
\begin{equation*}
		\begin{array}{llll}
			\hat{{\cal U}}: &   L^2(\Sigma,h (x,t) \dx \dy)  &  \to &  L^2(\Sigma) \\
			&    \psi   &  \mapsto  &   \sqrt{h}     \psi \end{array},
	\end{equation*}
we define 	\begin{align*}
		\hat{Q}_{f',g'}(\psi)  := &   \;Q_{f',g'}(\hat{{\cal U}}^{-1} \psi)  
		= Q_{f',g'}(h^{-1/2} \psi)  \\ 
	= & \int_{\Sigma} \frac{\left|(h^{-1/2} \psi)' - (f' \xi'_1+ g' \xi'_2) \partial_t (h^{-1/2} \psi) \right|^2}{h} \dx \dt  + \int_{\Sigma} 
 h \left| \partial_t (h^{-1/2} \psi) \right|^2 \dx \dt,
	\end{align*}
	\[\dom \hat{Q}_{f',g'}:= \dom Q_{f',g'},\]
	where 
	$(h^{-1/2} \psi)' := \partial (h^{-1/2} \psi) / \partial x$.
 Note that
\begin{align*}
    (h^{-1/2} \psi)' &=   ((1+ (f'(x)\xi'_2(t)-g'(x)\xi'_1(t))^2)^{-1/4} \psi)' =  h^{-1/2} \left[  \psi' - \frac{\partial_x h^2}{4h^2}  \psi \right], %\\
   % & = \partial_x ((1+ (f'(x)\xi'_2(t)-g'(x)\xi'_1(t))^2)^{-1/4}) \psi + (1+ (f'(x)\xi'_2(t)-g'(x)\xi'_1(t))^2)^{-1/4} \psi'\\
   % & = \left(-\frac{1}{2} h^{-5/2}(f'(x)\xi'_2(t)-g'(x)\xi'_1(t))(f''(x)\xi'_2(t)-g''(x)\xi'_1(t))\right) \psi + h^{-1/2} \psi'\\
   % & = h^{-1/2} \left[  \psi' - \frac{1}{2}  h^{-2}(f'(x)\xi'_2(t)-g'(x)\xi'_1(t))(f''(x)\xi'_2(t)-g''(x)\xi'_1(t)) \psi \right]\\
   % & =  h^{-1/2} \left[  \psi' - \frac{\partial_x h^2}{4h^2}  \psi \right],
\end{align*}
and
\begin{align*}
\partial_t(h_\beta^{-1/2} \psi) & =  \partial_t((1+ (f'(x)\xi'_2(t)-g'(x)\xi'_1(t))^2)^{-1/4} \psi) = h^{-1/2}  \left[ \partial_t \psi -\frac{\partial_t h^2}{4h^{2}} \psi \right].%\\
%& = \partial_t((1+ (f'(x)\xi'_2(t)-g'(x)\xi'_1(t))^2)^{-1/4}) \psi +(1+ (f'(x)\xi'_2(t)-g'(x)\xi'_1(t))^2)^{-1/4}  \partial_t \psi\\
%& =  \left( -\frac{1}{2}h^{-5/2} (f'(x)\xi'_2(t)-g'(x)\xi'_1(t))(f'(x)\xi''_2(t)-g'(x)\xi''_1(t))\right) \psi +h^{-1/2}  \partial_t \psi\\
%& = h^{-1/2}  \left[ \partial_t \psi -\frac{1}{2}h^{-2} (f'(x)\xi'_2(t)-g'(x)\xi'_1(t))(f'(x)\xi''_2(t)-g'(x)\xi''_1(t)) \psi \right]\\
%& = h^{-1/2}  \left[ \partial_t \psi -\frac{\partial_t h^2}{4h^{2}} \psi \right].
\end{align*}

Then, 
\begin{align*}
    \hat{Q}_{f',g'}(\psi) 
    %= & \int_{\Sigma} \frac{\left|h^{-1/2} \left[  \psi' - \frac{\partial_x h^2}{4h^2}  \psi \right] - (f' \xi'_1+ g' \xi'_2) h^{-1/2}\left[\partial_t \psi - \frac{ \partial_t h^2}{4 h^2} \psi  \right] \right|^2}{h} \dx \dt \\
    % & + \int_{\Sigma}  h \left| h^{-1/2}\left[\partial_t \psi - \frac{ \partial_t h^2}{4 h^2} \psi \right] \right|^2 \dx \dt\\
 & =  \int_{\Sigma} \frac{1}{h^2} \left| \psi' - \frac{\partial_x h^2}{4h^2}  \psi  - (f' \xi'_1+ g' \xi'_2)\left[\partial_t \psi - \frac{ \partial_t h^2}{4 h^2} \psi  \right] \right|^2 \dx \dt + \int_{\Sigma} 
 \left|\partial_t \psi - \frac{ \partial_t h^2}{4 h^2} \psi  \right|^2 \dx \dt. %\\
 % = & \int_{\Sigma}  \Bigg[  \frac{|\psi'|^2}{h^2} -  \frac{\partial_x h^2}{2h^4}  \psi \psi'  - \frac{2(f' \xi'_1+ g' \xi'_2)}{h^2} \left(\psi' \partial_t \psi  - \frac{ \partial_t h^2}{4 h^2} \psi \psi' \right) + \left (\frac{\partial_x h^2}{4h^3} \right)^2  |\psi|^2  \\
  % & +  \frac{(1+(f')^2 + (g')^2)}{h^2} \left|\partial_t \psi - \frac{ \partial_t h^2}{4 h^2} \psi  \right|^2  + (f' \xi'_1+ g' \xi'_2) \left( \frac{ \partial_x h^2}{2 h^4} \psi \partial_t \psi -  \frac{ \partial_x h^2 \partial_t h^2}{8 h^6} |\psi|^2\right) \Bigg] \dx \dt
\end{align*}

 Since, $f', g' \in L^\infty(\mathbb{R})$, one has $ \dom \hat{Q}_{f',g'} = {\cal H}^1(\Sigma)$.   Denote by $\hat{H}_{f',g'}$ the self-adjoint operator associated with the quadratic form $\hat{Q}_{f',g'}(\psi)$. Therefore, since $H_{f',g'}$ and $\hat{H}_{f',g'}$ are unitarily equivalent, we can identify  $-\Delta_{\cal S}$ with the self-adjoint operator $\hat{H}_{f',g'}$.

\section{Essential spectrum} \label{sec3}

This section concerns the proof of Theorem \ref{theo:surfess_surf}.  Inspired by \cite{briet}, we examine the situation in which $\beta_1$ and $\beta_2$ take finite values.  
 
	\begin{Lemma}\label{lemma:tiafg_surf}
	A real number $\lambda$ belongs to the essential spectrum of $\hat{H}_{f',g'}$ if, and only if, there exists a sequence $\{\psi_n\}_{n \in \mathbb{N}} \subset \dom \hat{Q}_{f',g'}$ such that the following three conditions hold:
 \begin{enumerate} [label=(\roman*)]
			\item $\|\psi_n\|_{L^2(\cal S)}=1$, $\forall n \in \mathbb{N}$;
			\item $(\hat{H}_{f',g'} -\lambda \Id) \psi_n \rightarrow 0$, as $n \rightarrow \infty$, in $(\dom \hat{H}_{f',g'})^{*}$,
      	\item $\operatorname{supp} \psi_n \subset {\cal S}\setminus (-n,n) \times {\cal C}$, $\forall n \in \mathbb{N}$.
		\end{enumerate}
	\end{Lemma}
	\begin{proof}
By the general Weyl criterion modified to quadratic forms, $\lambda \in \sigma_{ess}(\hat{H}_{f',g'})$ if, and only if, there exists a sequence $\{\upsilon_n\}_{n \in \mathbb{N}} \subset \dom \hat{Q}_{f',g'}$ such that $(i)$ and $(ii)$ hold but $(iii)$ is replaced by $(iii')$ $\upsilon_n \overset{w}{\rightarrow} 0$, as $n \rightarrow \infty$ in $L^2(\cal S)$

The sequence  $\{\psi_n\}_{n \in \mathbb{N}}$ satisfying $(i)$ and $(iii)$ clearly converges weakly to zero. Hence, one implication of the lemma is obvious. Conversely, let $\lambda \in \sigma_{ess}(\hat{H}_{f',g'})$. By the quadratic form criterion, there exists a sequence $\{\upsilon_n\}_{n \in \mathbb{N}} \subset \dom \hat{Q}_{f',g'}$ satisfying $(i)$, $(ii)$ and $(iii')$. Using this, we will construct a sequence $\{\psi_n\}_{n \in \mathbb{N}}$ satisfying $(i)$, $(ii)$ and $(iii)$. 

Note that in the condition $(ii)$,  $(\dom \hat{H}_{f',g'})^{*}$ denotes the dual space of $\dom (\hat{Q}_{f',g'})$. Furthermore, recall that condition $(ii)$ means that 
		\begin{equation}\label{dualf_surf}
			  \|(\hat{H}_{f',g'} - \lambda \Id) \psi_n\|_{-1} := \|(\hat{H}_{f',g'} - \lambda \Id) \psi_n\|_{(\dom \hat{Q}_{f',g'})^{*}}= \sup_{\phi \in \dom \hat{Q}_{f',g'} \setminus \{0\}} \frac{|\langle \phi, (\hat{H}_{f',g'} - \lambda \Id) \psi_n \rangle|}{\|\phi\|_1} \rightarrow 0, 
		\end{equation}
		as $n \rightarrow \infty $, where
		\begin{equation*}
			\|\phi\|_1^2 := \hat{Q}_{f',g'} (\phi) +\|\phi\|_{L^2(\cal S)}^{2}.
		\end{equation*}
		
By writing
\begin{equation}\label{eq:dualfg_surf}
    \upsilon_n = (\hat{H}_{f',g'} + \Id)^{-1} (\hat{H}_{f',g'} - \lambda \Id) \upsilon_n + (1+\lambda) (\hat{H}_{f',g'} + \Id)^{-1} \upsilon_n,
\end{equation}
and using  $(ii)$, we observe that the sequence  
 $\{\upsilon_n\}_{n \in \mathbb{N}}$ is bounded in $\dom \hat{Q}_{f',g'}$. 

  Let $\varphi \in C^\infty(\mathbb{R})$ be a real function such that  $\varphi(x) =0$, if $x \in [-1,1]$ and $\varphi(x)=1$,  if $x \in \mathbb{R}\setminus (-2,2)$. Define, for each $k \in \mathbb{N}$,
  \begin{equation*}
      \varphi_k(x) = \varphi\left(\frac{x}{k}\right);
  \end{equation*}
and we keep the same notation for the function  $\varphi_k \otimes \Id$ on $\mathbb{R} \times {\cal C}$, as well as,
for their derivatives $\varphi'$ and $\varphi''$. Note that $\operatorname{supp} \varphi_k \subset {\cal S} \setminus (-k,k) \times {\cal C}$.

For each $k \in \mathbb{N}$, the operator $(1-\varphi_k)(\hat{H}_{f',g'} + \Id)^{-1}$ is compact in $L^2(\cal S)$, hence, $\lim_{n \to \infty}( 1-\varphi_k)(\hat{H}_{f',g'} + \Id)^{-1} \upsilon_n =0$ in $L^2(\cal S)$. Then, there exists a subsequence  $\{\upsilon_{n_k}\}_{k\in \mathbb{N}}$ of $\{\upsilon_{n}\}_{n\in \mathbb{N}}$ such that $\lim_{k \to \infty} (1-\varphi_k)(\hat{H}_{f',g'} + \Id)^{-1} \upsilon_{n_k} =0$ em $L^2(\cal S)$. Consequently, by \eqref{eq:dualfg_surf} and $(ii)$  one has  $\lim_{k \to \infty}(1-\varphi_k)\upsilon_{n_k} = 0$ in $L^2(\cal S)$. Thus, we may assume $\|\varphi_k \upsilon_{n_k}\| \geq 1/2$, for each $k \in \mathbb{N}$, and define
  \begin{equation*}
      \psi_k := \frac{\varphi_k \upsilon_{n_k}}{\|\varphi_k \upsilon_{n_k}\|} \in \dom \hat{Q}_{f',g'}.
  \end{equation*}
  Note that $\psi_k$ satisfies $(i)$ and $(iii)$. It remains to verify $(ii)$. 
  
	Now, for every $\phi \in \dom \hat{Q}_{f',g'} \setminus \{0\}$, one has
		\begin{equation*}
				|\langle \phi, (\hat{H}_{f',g'} - \lambda \Id)\psi_k \rangle| = |\hat{Q}_{f',g'}(\phi,\psi_k) - \lambda \langle \phi,\psi_k \rangle_{L^2(\cal S)}|. 
		\end{equation*}
	Without loss of generality, assume $\|\varphi_k \upsilon_{n_k}\| =1$.	Integrating by parts, 

  	\begin{align*} 
\hat{Q}_{f',g'}(\phi,\psi_k )  = &  \int_{\Sigma} \frac{1}{h^2 }\left( \phi' - \frac{\partial_x h^2}{4h^2}  \phi - (f' \xi'_1+ g' \xi'_2)\left[\partial_t \phi - \frac{ \partial_t h^2}{4 h^2} \phi  \right] \right) \left( \psi_k' - \frac{\partial_x h^2}{4h^2}  \psi_k - (f' \xi'_1+ g' \xi'_2)\left[\partial_t \psi_k - \frac{ \partial_t h^2}{4 h^2} \psi_k  \right] \right)  \dx \dt\\
&+ \int_{\Sigma} 
 \left( \partial_t \phi - \frac{\partial_t h^2}{4 h^2} \phi  \right) \left( \partial_t \psi_k - \frac{\partial_t h^2}{4 h^2} \psi_k  \right) \dx \dt\\
 = & \; \hat{Q}_{f',g'}(\phi \varphi_k,\upsilon_{n_k}) + 2 \left\langle  \frac{1}{h} \left(\phi' - \frac{\partial_x h^2}{4h^2}  \phi - (f' \xi'_1+ g' \xi'_2)\left[\partial_t \phi - \frac{ \partial_t h^2}{4 h^2} \phi  \right]\right), \frac{1}{h}\varphi'_k \upsilon_{n_k} \right\rangle_{L^2(\Sigma)}\\
 & +  \left\langle  \frac{1}{h^2}  \phi, \varphi''_k \upsilon_{n_k} \right\rangle_{L^2(\Sigma)}  - \left \langle \left(
  \frac{\partial_x h^2}{2h^4} -(f'\xi'_1 + g'\xi'_2)\frac{\partial_t h^2}{2h^4} + \frac{(f'\xi''_1 + g'\xi''_2)}{h^2} \right)
 \phi, \varphi'_k \upsilon_{n_k} \right\rangle_{L^2(\Sigma)}, 
   \end{align*}
   and
  $ - \lambda \langle \phi,\psi_n \rangle_{L^2(\Sigma)}  = 
 - \lambda \langle \phi \varphi_k,\upsilon_{n_k} \rangle_{L^2(\Sigma)}$, for every test function $\varphi \in C_0^\infty(\Sigma)$. From (\ref{dualf_surf}), it follows that
{\small
\begin{align*}
     \sup_{\phi \in  C_0^\infty(\Sigma)\setminus \{0\}} \frac{|\hat{Q}_{f',g'}(\phi \varphi_k,\upsilon_{n_k})  - \lambda \langle \phi \varphi_k,\upsilon_{n_k} \rangle_{L^2(\Sigma)}|}{\|\phi\|_1} \leq &  \sup_{\substack{\phi \in   C_0^\infty(\Sigma) \setminus \{0\}\\ \phi \varphi_k \neq 0}} \frac{|\hat{Q}_{f',g'} (\phi \varphi_k,\upsilon_{n_k})  - \lambda \langle \phi \varphi_k,\upsilon_{n_k} \rangle_{L^2(\Sigma)}|}{\|\phi \varphi_k\|_1} \\
     = & \|(\hat{H}_{f',g'} - \lambda \Id) \upsilon_{n_k}\|_{-1} \to 0,
\end{align*}}
 as $k \to \infty$. At the same time, using the Cauchy-Schwarz inequality and the estimate 
 \begin{equation*}
     \left\|\frac{1}{h} \left( \phi' - \frac{\partial_x h^2}{4h^2}  \phi - (f' \xi'_1+ g' \xi'_2)\left[\partial_t \phi - \frac{ \partial_t h^2}{4 h^2} \phi  \right] \right) \right\|^2 \leq \hat{Q}_{f',g'} (\phi) \leq \|\phi\|_1^2, 
 \end{equation*}
 we get 
 \begin{equation*}
         \sup_{\phi \in   C_0^\infty(\Sigma) \setminus \{0\}} \frac{\left| 2 \left\langle  \frac{1}{h} \left(\phi' - \frac{\partial_x h^2}{4h^2}  \phi - (f' \xi'_1+ g' \xi'_2)\left[\partial_t \phi - \frac{ \partial_t h^2}{4 h^2} \phi  \right]\right), \frac{1}{h}\varphi'_k \upsilon_{n_k} \right\rangle_{L^2(\Sigma)} \right|}{\|\phi\|_1} \leq \|\varphi'_k\|_{\infty} \|\upsilon_{n_k}\| \frac{2}{\inf|h|},
 \end{equation*}
 where $\|\varphi'_k\|_{\infty}$ denotes the supremum norm of $\varphi'_k$. By the normalization $(i)$ and since $\|\varphi'_k\|_\infty = k^{-1}\|\varphi'\|_\infty$, we see that also this terms to zero as $k \to \infty$. Moreover, using the estimate $\|\phi\| \leq \|\phi\|_1$, and given that $\|\varphi''_k\|_\infty = k^{-2}\|\varphi''\|_\infty$, we conclude that
 \begin{equation*}
         \sup_{\phi \in   C_0^\infty(\Sigma) \setminus \{0\}} \frac{\left|  \left\langle  \frac{1}{h^2} \phi , \varphi''_k \upsilon_{n_k} \right\rangle_{L^2(\Sigma)} \right|}{\|\phi\|_1} \leq \|\varphi''_k\|_{\infty} \|\upsilon_{n_k}\| \frac{1}{\inf|h^2|} \to 0,
 \end{equation*}
and
\begin{equation*}
         \sup_{\phi \in  C_0^\infty(\Sigma) \setminus \{0\}} \frac{\left|  \left \langle \left(
  \frac{\partial_x h^2}{2h^4} -(f'\xi'_1 + g'\xi'_2)\frac{\partial_t h^2}{2h^4} + \frac{(f'\xi''_1 + g'\xi''_2)}{h^2} \right)
 \phi, \varphi'_k \upsilon_{n_k}  \right\rangle_{L^2(\Sigma)} \right|}{\|\phi\|_1} \leq \|\varphi'_k\|_{\infty} \|\upsilon_{n_k}\| K \to 0,
 \end{equation*}
as $k \to 0$, where in the last inequality  $K:= \sup\left|
  \frac{\partial_x h^2}{2h^4} -(f'\xi'_1 + g'\xi'_2)\frac{\partial_t h^2}{2h^4} + \frac{(f'\xi''_1 + g'\xi''_2)}{h^2} \right|$.

In summary, we have verified that $\|(\hat{H}_{f',g'} - \lambda \Id) \psi_k \|_{-1} \to 0$ as $k \to \infty$. Therefore, the sequence $\{\psi_k\}_{n \in \mathbb{N}}$ satisfies condition $(ii)$, which completes the proof.  
	\end{proof}

 Using the previous lemma, we immediately obtain the following result (which states that the essential spectrum of the Laplace operator is determined  by the surface's behaviour at infinity only)

\begin{Proposition}\label{prop:essfg_surf}
	Suppose that the conditions in  \eqref{eq:condlim_surf} hold. Then, $\sigma_{ess}(\hat{H}_{f',g'}) = \sigma_{ess}(\hat{H}_{\beta_1,\beta_2})$.
	\end{Proposition}

\begin{proof}
    Let $\lambda \in \sigma_{ess}(\hat{H}_{\beta_1,\beta_2})$. By Lemma \ref{lemma:tiafg_surf} , there exists a sequence $\{\psi_n\}_{n \in \mathbb{N}} \subset \dom \hat{H}_{\beta_1,\beta_2}$ satisfying the properties $(i)-(iii)$. By \eqref{eq:dualfg_surf}, for every $n \in \mathbb{N}$, we write
    \begin{equation*}
         \psi_n = (\hat{H}_{\beta_1, \beta_2} + \Id)^{-1} (\hat{H}_{\beta_1, \beta_2} - \lambda \Id) \psi_n + (1+\lambda) (\hat{H}_{\beta_1, \beta_2} + \Id)^{-1} \psi_n.
    \end{equation*}
    By $(ii)$ of Lemma \ref{lemma:tiafg_surf}, we can see that the sequence $\{\psi_n\}_{n \in \mathbb{N}}$ is bounded in  $\dom \hat{Q}_{\beta_1, \beta_2}$.

Let
  	\begin{align*} 
\hat{Q}_{f',g'}(\phi,\psi_n )  = &  \int_{\Sigma} \frac{1}{h^2 }\left( \phi' - \frac{\partial_x h^2}{4h^2}  \phi - (f' \xi'_1+ g' \xi'_2)\left[\partial_t \phi - \frac{ \partial_t h^2}{4 h^2} \phi  \right] \right) \left( \psi_n' - \frac{\partial_x h^2}{4h^2}  \psi_n - (f' \xi'_1+ g' \xi'_2)\left[\partial_t \psi_n - \frac{ \partial_t h^2}{4 h^2} \psi_n  \right] \right)  \dx \dt\\
&+ \int_{\Sigma} 
 \left( \partial_t \phi - \frac{\partial_t h^2}{4 h^2} \phi  \right) \left( \partial_t \psi_n - \frac{\partial_t h^2}{4 h^2} \psi_n  \right) \dx \dt,
 \end{align*}
 and
  	\begin{align*} 
\hat{Q}_{\beta_1,\beta_2}(\phi,\psi_n )  = &  \int_{\Sigma} \frac{1}{h_{\beta_1,\beta_2}^2 }\left( \phi' - (\beta_1 \xi'_1+ \beta_2 \xi'_2)\left[\partial_t \phi - \frac{ (h_{\beta_1,\beta_2}^2)'}{4 h_{\beta_1,\beta_2}^2} \phi  \right] \right) \left( \psi'_n  - (\beta_1 \xi'_1+ \beta_2 \xi'_2)\left[\partial_t \psi_n - \frac{  (h_{\beta_1,\beta_2}^2)'}{4 h_{\beta_1,\beta_2}^2} \psi_n  \right] \right)  \dx \dt\\
&+ \int_{\Sigma} 
 \left( \partial_t \phi - \frac{(h_{\beta_1,\beta_2}^2)'}{4 h_{\beta_1,\beta_2}^2} \phi  \right) \left( \partial_t \psi_n - \frac{ (h_{\beta_1,\beta_2}^2)'}{4 h_{\beta_1,\beta_2}^2} \psi_n  \right) \dx \dt.
 \end{align*}

Some calculations show that 
\begin{equation}\label{ident:2_surf}
    \begin{split}
  \hat{Q}_{f',g'}(\phi, \psi_n)   = & \;\hat{Q}_{\beta_1,\beta_2}(\phi, \psi_n)
    + \int_\Sigma \Bigg(\left[\frac{1}{h^2} - \frac{1}{h_{\beta_1, \beta_2}^2} \right] \phi' \psi'_n  + \left[\frac{1+(f')^2 +(g')^2}{h^2} - \frac{1+\beta_1^2 +\beta_2^2}{h_{\beta_1, \beta_2}^2}  \right] \partial_t \phi \partial_t  \psi_n\\
    & + \left[(f' \xi'_1+ g'\xi'_2)  \frac{\partial_t h^2}{4h^4} - (\beta_1 \xi'_1+ \beta_2\xi'_2)  \frac{ (h_{\beta_1, \beta_2}^2)'}{4h_{\beta_1, \beta_2}^4}   \right](\phi' \psi_n +\phi \psi'_n) \\
    & - \left[\frac{f' \xi'_1+ g'\xi'_2}{h^2} -\frac{\beta_1 \xi'_1+ \beta_2\xi'_2}{h_{\beta_1, \beta_2}^2}\right](\phi' \partial_t  \psi_n + \partial_t \phi \psi'_n) \\
    & - \left[(1+(f')^2 +(g')^2) \frac{\partial_t h^2}{4h^4} - (1+\beta_1^2 +\beta_2^2) \frac{( h_{\beta_1, \beta_2}^2)'}{4h_{\beta_1, \beta_2}^4} \right] (  \partial_t\phi \psi_n + \phi \partial_t \psi_n)  \\
    &   + \left[(1+(f')^2 +(g')^2) \left(\frac{\partial_t h^2}{4h^3} \right)^2 - (1+\beta_1^2 +\beta_2^2) \left(\frac{( h_{\beta_1, \beta_2}^2)'}{4h_{\beta_1, \beta_2}^3}\right)^2\right]  \phi \psi_n \Bigg) \dx \dt\\
    & - \left\langle  \frac{1}{h} \left(\phi' - \frac{\partial_x h^2}{4h^2}  \phi - (f' \xi'_1+ g' \xi'_2)\left[\partial_t \phi - \frac{ \partial_t h^2}{4 h^2} \phi  \right]\right),   \frac{\partial_x h^2}{4h^3} \psi_{n} \right\rangle_{L^2(\Sigma)}\\
    & - \left\langle \phi , \frac{\partial_x h^2}{4h^4} \left(\psi'_n - (f' \xi'_1+ g' \xi'_2)\left[\partial_t \psi_n - \frac{\partial_t h^2}{4 h^2} \psi_n  \right]\right)\right\rangle_{L^2(\Sigma)}.
\end{split}
\end{equation}
Since  $\{\psi_n\}_{n \in \mathbb{N}}$ is bounded in  $\dom \hat{Q}_{\beta_1, \beta_2}$,  $\|\phi\|_1^2 := \hat{Q}_{f',g'} (\phi) +\|\phi\|_{L^2(\cal S)}^{2}$, using the Cauchy-Schwarz inequality and the fact that $\| \partial_x h^2 / 4 h^i\|_{L^\infty(\Sigma \setminus (-n,n) \times {\cal C})} \to 0$,  as $n \to \infty$, para $i=3,4$, we obtain  
\begin{align*}
         & \sup_{\substack{\phi \in  \dom \hat{Q}_{f',g'} \\ \phi \neq 0}}    
          \frac{\left| \left\langle  \frac{1}{h} \left(\phi' - \frac{\partial_x h^2}{4h^2}  \phi - (f' \xi'_1+ g' \xi'_2)\left[\partial_t \phi - \frac{ \partial_t h^2}{4 h^2} \phi  \right]\right),  \frac{\partial_x h^2}{4h^3} \psi_{n} \right\rangle_{L^2(\Sigma)} \right|}{\|\phi\|_1} \leq  \left\| \frac{\partial_x h^2}{4h^3} \right\|_{L^\infty(\Sigma \setminus (-n, n) \times {\cal C})} \|\psi_{n}\| \to 0,
 \end{align*}
 and 
\begin{align*}
         & \sup_{\substack{\phi \in  \dom \hat{Q}_{f',g'} \\ \phi \neq 0}}    \frac{\left| \left\langle \phi , \frac{\partial_x h^2}{4h^4} \left(\psi'_n - (f' \xi'_1+ g' \xi'_2)\left[\partial_t \psi_n - \frac{\partial_t h^2}{4 h^2} \psi_n  \right]\right)\right\rangle_{L^2(\Sigma)} \right|}{\|\phi\|_1} \\
         & \qquad  \leq  \left\| \frac{\partial_x h^2}{4h^4} \right\|_{L^\infty(\Sigma \setminus (-n, n) \times {\cal C})} 
 \left( \|\psi'_n\| + \|(f' \xi'_1+ g' \xi'_2)\|_{L^\infty(\Sigma)} \left[\|\partial_t \psi_n\| +\left\| \frac{\partial_t h^2}{4 h^2}\right\|_{L^\infty(\Sigma)} \|\psi_n\|  \right] \right) \to 0,
 \end{align*}
 as $n \to \infty$. Moreover, observe that 
\begin{gather*}
\left\| \frac{1}{h^2} - \frac{1}{h_{\beta_1, \beta_2}^2} \right\|_{L^\infty(\Sigma \setminus (-n, n) \times {\cal C})}  \to 0, \quad \left\| \frac{1+(f')^2 +(g')^2}{h^2} - \frac{1+\beta_1^2 +\beta_2^2}{h_{\beta_1, \beta_2}^2} \right\|_{L^\infty(\Sigma \setminus (-n, n) \times {\cal C})}  \to 0, \\
\left\| (f' \xi'_1+ g'\xi'_2)  \frac{\partial_t h^2}{4h^4} - (\beta_1 \xi'_1+ \beta_2\xi'_2)  \frac{ (h_{\beta_1, \beta_2}^2)'}{4h_{\beta_1, \beta_2}^4} \right\|_{L^\infty(\Sigma \setminus (-n, n) \times {\cal C})}  \to 0, \\
\left\| \frac{f' \xi'_1+ g'\xi'_2}{h^2} -\frac{\beta_1 \xi'_1+ \beta_2\xi'_2}{h_{\beta_1, \beta_2}^2} \right\|_{L^\infty(\Sigma \setminus (-n, n) \times {\cal C})} \to 0, \\
\left\| (1+(f')^2 +(g')^2) \frac{\partial_t h^2}{4h^4} - (1+\beta_1^2 +\beta_2^2) \frac{( h_{\beta_1, \beta_2}^2)'}{4h_{\beta_1, \beta_2}^4} \right\|_{L^\infty(\Sigma \setminus (-n, n) \times {\cal C})}  \to 0,\\
\left\| (1+(f')^2 +(g')^2) \left(\frac{\partial_t h^2}{4h^3} \right)^2 - (1+\beta_1^2 +\beta_2^2) \left(\frac{( h_{\beta_1, \beta_2}^2)'}{4h_{\beta_1, \beta_2}^3}\right)^2\right\|_{L^\infty(\Sigma \setminus (-n, n) \times {\cal C})}  \to 0,
 \end{gather*}
as $n \to \infty$. Due to the fact that $\operatorname{supp}  \psi'_n, \, \operatorname{supp}  \partial_t \psi_n \subseteq \operatorname{supp}  \psi_n$ and $\|(\hat{H}_{\beta_1, \beta_2} - \lambda \Id) \psi_n\|_{-1} \to 0$, as $n \to \infty$,  by $(ii)$ of Lemma \ref{lemma:tiafg_surf}, it follows from \eqref{ident:2_surf} and \eqref{dualf_surf} that $\|(\hat{H}_{f',g'} - \lambda \Id) \psi_n\|_{-1} \to 0$, as $n \to \infty$. Therefore, $\lambda \in \sigma_{ess}(\hat{H}_{f',g'})$. Analogously, we obtain the inclusion   $\sigma_{ess}(\hat{H}_{f',g'}) \subset \sigma_{ess}(\hat{H}_{\beta_1,\beta_2})$.
\end{proof}

In view of Proposition \ref{prop:essfg_surf}, we are going to the study the essential spectrum of the operator $\hat{H}_{\beta_1, \beta_2}$ instead of $\hat{H}_{f', g'}$. It remains to determine the essential spectrum of the operator  $\hat{H}_{\beta_1, \beta_2}$. 

\begin{Proposition} \label{prop:essfg2_surf}
		The spectrum of the operator $\hat{H}_{\beta_1, \beta_2}$ is purely essential and equals the interval  $[E_1(0), \infty)$.
	\end{Proposition}
\begin{proof}
    Let ${\cal F}_x : L^2(\Sigma) \longrightarrow L^2(\Sigma)$ be the partial Fourier transform in the longitudinal variable $x$. ${\cal F}_x$ is a unitary
operator and, for functions $\psi \in L^1(\Sigma)$, the explicit expression for this transformation is given by 
    \begin{equation*}
        ({\cal F}_x \psi)(p,t) = \frac{1}{\sqrt{2 \pi}} \int_{\mathbb{R}} e^{-ipx} \psi(x,t) \dx.
    \end{equation*}
    We consider the operator $\tilde{H}_{\beta_1,\beta_2} := {\cal F}_x \hat{H}_{\beta_1,\beta_2} {\cal F}_x^{-1}$ which admits a direct integral decomposition  
    \begin{equation*}
        \tilde{H}_{\beta_1,\beta_2} = \int_{\mathbb{R}}^{\oplus} \hat{H}_{\beta_1,\beta_2}(p) {\rm d}p,
    \end{equation*}
    where, for every $p \in \mathbb{R}$, in the distributional sense
    \begin{align*}
        \hat{H}_{\beta_1,\beta_2}(p)  =& 
        -\frac{1}{h_{\beta_1,\beta_2}^2} (ip)^2  + 2ip  \frac{(\beta_1 \xi'_1 + \beta_2 \xi'_2)}{h_{\beta_1,\beta_2}^2}\partial_t  - \frac{\partial}{\partial t} \left[ \left(\frac{1+\beta_1^2 + \beta_2^2}{h_{\beta_1,\beta_2}^2}\right) \partial_t \right] + (1+\beta_1^2 + \beta_2^2) \left[\left(\frac{(h_{\beta_1,\beta_2}^2)'}{4 h_{\beta_1,\beta_2}^4} \right)' + \left( \frac{(h_{\beta_1,\beta_2}^2)'}{4 h_{\beta_1,\beta_2}^3}\right)^2\right], 
    \end{align*}
    $\dom \hat{H}_{\beta_1, \beta_2}(p) = \{v \in {\cal H}^1({\cal C}): \hat{H}_{\beta_1, \beta_2}(p) v \in L^2({\cal C}), \; v'(0)=v'(1)\}$; if $p=0$ we obtain the operator $T_{\beta_1, \beta_2}$ given by \eqref{fourpp:surf_surf} in the Introduction. Since 
    \begin{equation}\label{uniesp_surf}
        \sigma(\hat{H}_{\beta_1, \beta_2}) = \bigcup_{p \in \mathbb{R}} \sigma(\hat{H}_{\beta_1, \beta_2}(p)),
    \end{equation}
the strategy is to study the family of operators $\{\hat{H}_{\beta_1, \beta_2}(p) : p \in \mathbb{R}\}$. By the compactness of the embedding  ${\cal H}^1 ({\cal C}) \hookrightarrow L^2({\cal C})$, each operator $\hat{H}_{\beta_1, \beta_2}(p)$ has a purely discrete spectrum. In particular, denote by   $ \{E_n(0)\}_{n \in \mathbb{N}}$  the non-decreasing sequence of the eigenvalues of  $\hat{H}_{\beta_1, \beta_2}(0)=T_{\beta_1, \beta_2}$. Note that 
 \begin{align*}
        \hat{H}_{\beta_1,\beta_2}(p)  =& 
        - \frac{1+\beta_1^2 + \beta_2^2}{h_{\beta_1,\beta_2}^2} \left( \partial_t - i \frac{p(\beta_1 \xi'_1 + \beta_2 \xi'_2)}{1+\beta_1^2 + \beta_2^2}\right)^2 - \frac{p^2 (\beta_1 \xi'_1 + \beta_2 \xi'_2)^2}{h_{\beta_1,\beta_2}^2 (1+\beta_1^2 + \beta_2^2)} + \frac{p^2}{h_{\beta_1,\beta_2}^2}\\
        & -  \left(\frac{1+\beta_1^2 + \beta_2^2}{h_{\beta_1,\beta_2}^2}\right)' \partial_t  + (1+\beta_1^2 + \beta_2^2) \left[\left(\frac{(h_{\beta_1,\beta_2}^2)'}{4 h_{\beta_1,\beta_2}^4} \right)' + \left( \frac{(h_{\beta_1,\beta_2}^2)'}{4 h_{\beta_1,\beta_2}^3}\right)^2\right]\\
        = & - \frac{1+\beta_1^2 + \beta_2^2}{h_{\beta_1,\beta_2}^2} \left( \partial_t - i \frac{p(\beta_1 \xi'_1 + \beta_2 \xi'_2)}{1+\beta_1^2 + \beta_2^2}\right)^2  + \frac{p^2}{1+\beta_1^2 + \beta_2^2}\\
        & -  \left(\frac{1+\beta_1^2 + \beta_2^2}{h_{\beta_1,\beta_2}^2}\right)' \partial_t  + (1+\beta_1^2 + \beta_2^2) \left[\left(\frac{(h_{\beta_1,\beta_2}^2)'}{4 h_{\beta_1,\beta_2}^4} \right)' + \left( \frac{(h_{\beta_1,\beta_2}^2)'}{4 h_{\beta_1,\beta_2}^3}\right)^2\right]\\
        = &  \; e^{i \gamma(p)} t  \left[ T_{\beta_1, \beta_2}    + \frac{p^2}{1+\beta_1^2 + \beta_2^2} \right] e^{-i \gamma(p)} t,
    \end{align*}
    where
    \begin{equation*}
        \gamma(p) := \frac{p(\beta_1 \xi'_1 + \beta_2 \xi'_2)}{1+\beta_1^2 + \beta_2^2}.  
    \end{equation*}
As a consequence, 
\begin{equation}\label{sigmap_surf}
    \sigma(\hat{H}_{\beta_1, \beta_2}(p)) = \left\{ E_n(0)    + \frac{p^2}{1+\beta_1^2 + \beta_2^2} \right\}_{n \in \mathbb{N}}.
\end{equation}
Therefore, by \eqref{uniesp_surf} and \eqref{sigmap_surf}, we obtain the desired result.

\end{proof}

{\bf Proof of Theorem \ref{theo:surfess_surf}:} 
 It remains to apply Propositions \ref{prop:essfg_surf} and \ref{prop:essfg2_surf}. \qed

	\section{Discrete spectrum}\label{sec4}

This section is dedicated to prove Theorems \ref{exidisspc_surf} and  \ref{exidisspc_surf2} stated in the introduction. 

At first, we fix some notations that we will be useful in this section. Let $Q$ be a closed and lower bounded sesquilinear form with domain  $\mathrm{dom}~ Q$ dense in a Hilbert space  $H$. Denote by $A$ the self-adjoint operator associated with $Q$. The Rayleigh quotients of $A$ can be defined by %
\begin{equation}\label{rayquoin}
	\lambda_j(A) = \inf_{\substack{G \subset \mathrm{dom}~ Q \\ \dim G=j}}   \sup_{\substack{ \psi \in G \\ \psi \neq 0}}  \frac{Q (\psi)}{ \|\psi \|^2_{H}}.
\end{equation}
Let $\mu = \inf \sigma_{ess} (A)$. The sequence $\{\lambda_j(A)\}_{j \in \mathbb{N}}$ is non-decreasing and satisfies: (i) if $\lambda_j(A) < \mu$, then it is a discrete eigenvalue of $A$; (ii) If $\lambda_j(A) \geq \mu$, then 
$\lambda_j(A) = \mu$; (iii) The $j$-th  eigenvalue of $A$ less than $\mu$ (it is exists) coincides with $\lambda_j(A)$.

\vspace{0.2cm}
	\noindent	
{\bf Proof of Theorem \ref{exidisspc_surf}:}
 	Consider the quadratic form  
	\begin{equation*}
		\hat{q}_{f',g'} (\psi) = \hat{Q}_{f',g'}(\psi) - E_1(0)\|\psi\|_{L^2(\Sigma)}^{2}, \quad \dom \hat{q}_{f',g'} = \dom \hat{Q}_{f',g'}.
	\end{equation*}
	According to (\ref{rayquoin}) and Theorem  \ref{theo:surfess_surf}, it is enough show that there exists a function  $\psi \in \dom \hat{q}_{f',g'} \setminus \{0\}$ such that $\hat{q}_{f',g'}(\psi) < 0$. 
	
	The first step is to find a sequence  $\{\psi_n\}_{n \in \mathbb{N}} \subset \dom \hat{q}_{f',g'}$ such that $\hat{q}_\beta (\psi_n) \to 0$ as $n \to \infty$. To that end, let $w \in C^\infty(\mathbb{R})$ be a real-valued function such that $w=1$ for $x \in [-1,1]$, and $w=0$ for $x \in \mathbb{R} \setminus [-2,2]$. For each every $n \in \mathbb{N}$, define
	\begin{equation*}
		w_n(x) = w \left(\frac{x}{n}\right) \quad \mbox{and} \quad \psi_n(x,t) = w_n(x) \chi (t),
	\end{equation*}  
	where $\chi$ denotes the normalized eigenfunction corresponding to the eigenvalue $E_1(0)$ of the two-dimensional operator $T_{\beta_1, \beta_2}$ given by \eqref{fourpp:surf_surf}.
 
 In particular, 
	\begin{equation}\label{estdericuoff1_surf}
		\int_{\mathbb{R}} |w'_n|^2 \dx = \frac{1}{n} \int_{\mathbb{R}} |w'|^2 \dx \to 0, \quad \mbox{as} \quad n \to \infty,
	\end{equation}
	and
 \begin{equation}\label{firsteigcross1_surf}
		\int_{\cal C} 
 (1+\beta_1^2 + \beta_2^2) \left( \frac{|\chi'|^2}{h_{\beta_1, \beta_2}^2}  +  \left[
  \left(\frac{ (h_{\beta_1, \beta_2}^2)'}{4 h_{\beta_1, \beta_2}^4} \right)' + 
 \left(\frac{ (h_{\beta_1, \beta_2}^2)'}{4 h_{\beta_1, \beta_2}^3}\right)^2 \right] |\chi|^2 \right)  \dt = E_1(0).
	\end{equation}
	By (\ref{firsteigcross1_surf}), one has 
	\begin{equation*}
		\begin{split}
	\hat{q}_{f',g'}(\psi_n) = &  \; \hat{Q}_{f',g'}(\psi_n) - E_1(0)\|\psi_n\|_{L^2(\Sigma)}^{2}\\
= & \int_{\Sigma}  \Bigg[  \frac{|w'_n|^2 |\chi|^2}{h^2} -  \frac{\partial_x h^2}{2h^4}  w_n w'_n |\chi|^2  - \frac{2(f' \xi'_1+ g' \xi'_2)}{h^2} \left(w'_n \chi w_n  \chi'  - \frac{ \partial_t h^2}{4 h^2} w_n w'_n |\chi|^2 \right)   \\
  & + \left (\frac{\partial_x h^2}{4h^3} \right)^2  |w_n|^2 |\chi|^2 + (f' \xi'_1+ g' \xi'_2) \left( \frac{ \partial_x h^2}{2 h^4} |w_n|^2 \chi \chi' -  \frac{ \partial_x h^2 \partial_t h^2}{8 h^6} |w_n|^2 |\chi|^2\right) \\
  & +  (1+(f')^2 + (g')^2) \left( \frac{|w_n|^2 |\chi'|^2}{h^2}  +  \left[
  \frac{\partial}{\partial t} \left(\frac{\partial_t h^2}{4 h^4} \right) + 
 \left(\frac{ \partial_t h^2}{4 h^3}\right)^2 \right] |w_n|^2 |\chi|^2 \right) \Bigg] \dx \dt
 \\
  & - \int_{\Sigma}  
 (1+\beta_1^2 + \beta_2^2) \left( \frac{|w_n |^2 |\chi'|^2}{h_{\beta_1, \beta_2}^2}  +  \left[ \left(\frac{(h_{\beta_1, \beta_2}^2)'}{4 h_{\beta_1, \beta_2}^4} \right)' + 
 \left(\frac{ (h_{\beta_1, \beta_2}^2)'}{4 h_{\beta_1, \beta_2}^3}\right)^2 \right] |w_n |^2 |\chi|^2 \right)   \dx\dt.
 \end{split}
	\end{equation*}
Note that 
\begin{gather*}
     \int_{\cal C}  \frac{2(f' \xi'_1+ g' \xi'_2)}{h^2} \chi  \chi'\dt 
     = -\int_{\cal C}  \frac{\partial}{\partial t}\left(\frac{f' \xi'_1+ g' \xi'_2}{h^2} \right) |\chi|^2 \dt,\\
    \int_{\cal C}  (f' \xi'_1+ g' \xi'_2)\frac{\partial_x h^2}{2h^2} \chi  \chi'\dt  = -\int_{\cal C}  \frac{\partial}{\partial t}\left( (f' \xi'_1+ g' \xi'_2)\frac{\partial_x h^2}{4h^2} \right) |\chi|^2 \dt.
\end{gather*}
      Then,
 	\begin{equation*}
		\begin{split}
	\hat{q}_{f',g'}(\psi_n) = & \int_{\mathbb{R}}   A(x)|w'_n|^2 \dx + \int_{\mathbb{R}} B(x)  w_n w'_n \dx \\
 & +  \int_{\mathbb{R}} \left( C(x) + (1+(f')^2 + (g')^2) D(x) -  
 (1+\beta_1^2 + \beta_2^2)  E \right) |w_n |^2   \dx \\
  \leq &  \int_{\mathbb{R}}   A(x)|w'_n|^2 \dx + \left(\int_{\mathbb{R}} |B(x)|^2|w_n|^2 \dx \right)^{1/2} \left(  \int_{\mathbb{R}} |w^{\prime}_n|^2 \dx \right)^{1/2}  +  \int_{\mathbb{R}} V(x)|w_n |^2 \dx.
  \end{split}
	\end{equation*}
 By (\ref{estdericuoff1_surf}) and since $\|w_n\|_\infty \leq 1$ and $\int_\mathbb{R} V(x) \dx <0$,  by the dominated convergence theorem, it follows that
\begin{equation*}
    \hat{q}_{f',g'}(\psi_n) \to \int_{\mathbb{R}} V(x) \dx, \quad \mbox{as} \quad n \to \infty. 
\end{equation*}
Thus, there exists $N \in \mathbb{N}$ such that $\hat{q}_{f',g'}(\psi_N)<0$.  \qed

\vspace{0.2cm}
	\noindent
{\bf Proof of Theorem \ref{exidisspc_surf2}:}

  For every $n \in \mathbb{R}$, consider the functions $w_n$ and $\psi_n$ as defined in the proof of Theorem \ref{exidisspc_surf}.  However, we now introduce a small perturbation and define
	\begin{equation*}
		\psi_{n,\varepsilon}(x,t) := \psi_n(x,t) + \varepsilon \phi(x,t),  
	\end{equation*}
where $\varepsilon \in \mathbb{R}$ and $\phi \in \dom \hat{q}_\beta$. In this case,
	\begin{equation*}
		\hat{q}_{f',g'}(\psi_{n,\varepsilon}) = \hat{q}_{f',g'}(\psi_n) + 2 \varepsilon \,  {\rm Re} \, (\hat{q}_{f',g'}(\psi_n, \phi))
		+ \varepsilon^2 \hat{q}_{f',g'}(\phi).
	\end{equation*}
	
	The strategy is to show that there exists $\phi$ satisfying 
	\begin{equation}\label{rbetaneez:surf}
		\lim_{n \to \infty} \hat{q}_{f',g'}(\psi_n, \phi) \neq 0.
	\end{equation}	
	Indeed, if (\ref{rbetaneez:surf}) holds true, it is enough to chosen $\varepsilon$ such that $\hat{q}_{f',g'}(\psi_{n,\varepsilon})<0$, for some $n$ large enough. 
     
	Consider $\eta \in C^\infty_0 (\mathbb{R})$, with $\operatorname{supp}  \eta \subset (-1,1)$.  Define $\phi(x,t) = \eta(x) \chi(t)$. We get 
\begin{align*}
    \hat{q}_{f',g'}(\psi_n, \phi )  = &  \int_{\Sigma} \frac{1}{h^2 } \left( \psi_n' - \frac{\partial_x h^2}{4h^2}  \psi_n - (f' \xi'_1+ g' \xi'_2)\left[\partial_t \psi_n - \frac{ \partial_t h^2}{4 h^2} \psi_n  \right] \right) \left( \phi' - \frac{\partial_x h^2}{4h^2}  \phi - (f' \xi'_1+ g' \xi'_2)\left[\partial_t \phi - \frac{ \partial_t h^2}{4 h^2} \phi  \right] \right)  \dx \dt\\
&+ \int_{\Sigma} 
  \left( \partial_t \psi_n - \frac{\partial_t h^2}{4 h^2} \psi_n  \right) \left( \partial_t \phi - \frac{\partial_t h^2}{4 h^2} \phi  \right) \dx \dt - E_1(0) \int_{\Sigma} \psi_n \phi \, \dx \dt\\
 % = & \int_{\Sigma} \frac{1}{h^2 } \Bigg(  \psi_n'\phi' - \frac{\partial_x h^2}{4h^2}  \psi_n'\phi - (f' \xi'_1+ g' \xi'_2) \psi_n' \left[\partial_t \phi - \frac{ \partial_t h^2}{4 h^2} \phi  \right]  \\
 % & -  \frac{\partial_x h^2}{4h^2}  \psi_n \phi' + \frac{\partial_x h^2}{4h^2}  \frac{\partial_x h^2}{4h^2}  \psi_n \phi + (f' \xi'_1+ g' \xi'_2)  \frac{\partial_x h^2}{4h^2}  \psi_n\left[\partial_t \phi - \frac{ \partial_t h^2}{4 h^2} \phi  \right] \\
 % & - (f' \xi'_1+ g' \xi'_2)\left[\partial_t \psi_n \phi' - \frac{ \partial_t h^2}{4 h^2} \psi_n \phi' - \frac{ \partial_x h^2}{4 h^2} \partial_t \psi_n\phi + \frac{\partial_x h^2}{4h^2} \frac{ \partial_t h^2}{4 h^2} \psi_n \phi\right] \Bigg) \dx \dt\\
%&+ \int_{\Sigma} \frac{(1+(f')^2 + (g')^2)}{h^2} \left( \partial_t \psi_n - \frac{\partial_t h^2}{4 h^2} \psi_n  \right) \left( \partial_t \phi - \frac{\partial_t h^2}{4 h^2} \phi  \right) \dx \dt - E_1(0) \int_{\Sigma} \psi_n \phi\, \dx \dt\\
   = & \int_{\Sigma} \frac{1}{h^2 } \Bigg(  w'_n \eta' |\chi|^2  - \frac{\partial_x h^2}{4h^2}  w'_n \eta |\chi|^2 - (f' \xi'_1+ g' \xi'_2) w'_n \chi \left[\eta \chi' - \frac{ \partial_t h^2}{4 h^2} \eta \chi  \right]  \\
  & -  \frac{\partial_x h^2}{4h^2}  w_n  \eta' |\chi|^2 + \frac{\partial_x h^2}{4h^2}  \frac{\partial_x h^2}{4h^2}  w_n \eta |\chi|^2  + (f' \xi'_1+ g' \xi'_2)  \frac{\partial_x h^2}{4h^2}  w_n \chi \left[\eta \chi' - \frac{ \partial_t h^2}{4 h^2} \eta \chi  \right] \\
  & - (f' \xi'_1+ g' \xi'_2)\left[ w_n  \chi' \eta' \chi - \frac{ \partial_t h^2}{4 h^2} w_n  \eta' |\chi|^2 - \frac{ \partial_x h^2}{4 h^2} w_n \chi' \eta \chi + \frac{\partial_x h^2}{4h^2} \frac{ \partial_t h^2}{4 h^2} w_n \eta |\chi|^2 \right] \Bigg) \dx \dt\\
&+ \int_{\Sigma} \frac{(1+(f')^2 + (g')^2)}{h^2}
  \left( w_n \chi' - \frac{\partial_t h^2}{4 h^2} w_n \chi  \right) \left( \eta \chi' - \frac{\partial_t h^2}{4 h^2} \eta \chi   \right) \dx \dt - E_1(0) \int_{\Sigma} w_n  \eta |\chi|^2 \dx \dt\\
   \to &   \int_{\Sigma} \eta'\left[ -  \frac{\partial_x h^2}{4h^4}  + \frac{\partial}{\partial t}\left(\frac{f' \xi'_1+ g' \xi'_2}{2h^2} \right)  + (f' \xi'_1+ g' \xi'_2) \frac{ \partial_t h^2}{4 h^4}   \right] |\chi|^2 \dx \dt\\
& + \int_{\Sigma} \eta \left[\left(\frac{\partial_x h^2}{4h^3}\right)^2  - \frac{\partial}{\partial t}\left((f' \xi'_1+ g' \xi'_2)  \frac{\partial_x h^2}{4h^4} \right)    - (f' \xi'_1+ g' \xi'_2)  \frac{\partial_x h^2 \partial_t h^2}{8h^6} \right] |\chi|^2 \dx \dt  \\
&+ \int_{\Sigma} \frac{(1+(f')^2 + (g')^2)}{h^2} \eta
  \left( \chi' - \frac{\partial_t h^2}{4 h^2} \chi  \right)^2 \dx \dt - E(0) \int_{\Sigma} \eta |\chi|^2 \dx \dt\\
  = & \int_{\mathbb{R}} \eta' \frac{B(x)}{2} \dx  + \int_{\mathbb{R}} \eta V(x) \dx\\
  = & \int_{\mathbb{R}} \left(V(x)-\frac{B'(x)}{2}  \right) \dx,
\end{align*}
as $n \to \infty$. If the last integral vanishes for all $\eta \in C_0^\infty(\mathbb{R})$, we would have $ V(x) -B'(x)/2 = 0$ almost everywhere, which implies $B(x) = 2 \int_{\mathbb{R}} V(x) \dx + C = C$. This yields a contradiction since, by hypothesis, $B(x)$ is not constant. Therefore, there exists a function  $\phi$ satisfying \eqref{rbetaneez:surf}. \qed

\section{Broken sheared waveguides shaped surfaces} \label{sec5}

This section is devoted to the proof of Theorems \ref{propress_surfcant} and \ref{exidisspc_surfcant}. 

	Recall the mapping ${\cal P}_\beta$ given by \eqref{lmas_surfcant}. Namely, ${\cal S}_\beta = \mathcal{P}_\beta(\Sigma)$. Following the same change of coordinates approach as in Section \ref{sec3}, and since ${\cal P}_\beta$ is a global diffeomorphism between $\Sigma$ and ${\cal S}_\beta$, we have that
\begin{align*}
    \hat{Q}_{\beta }(\psi) 
 & =  \int_{\Sigma} \frac{1}{h_\beta^2} \left| \psi'  -  g'(x) \xi'_2(t) \left[\partial_t \psi - \frac{ (h_\beta^2)'}{4 h_\beta^2} \psi  \right] \right|^2 \dx \dt + \int_{\Sigma} 
 \left|\partial_t \psi - \frac{(h_\beta^2)'}{4 h_\beta^2} \psi  \right|^2 \dx \dt, \quad  \dom \hat{Q}_{\beta} = {\cal H}^1(\Sigma).  
\end{align*}
  Denote by $\hat{H}_{\beta}$ the self-adjoint operator associated with the quadratic form 
 $\hat{Q}_{\beta}(\psi)$. Therefore, since $-\Delta_{{\cal S}_\beta}$ and $\hat{H}_{\beta}$ are unitarily equivalent, we can identify $-\Delta_{{\cal S}_\beta}$ with the self-adjoint operator $\hat{H}_{\beta}$.

\vspace{0.2cm}
	\noindent
{\bf Proof of Theorem \ref{propress_surfcant}:} It remains to apply Theorems \ref{prop:essfg_surf} and  \ref{prop:essfg2_surf}  to the case where $\hat{H}_{\beta_1, \beta_2} = \hat{H}_{\beta}$, with $\beta_1^2 = 0$ and $\beta_2^2 = \beta$.   \qed

\vspace{0.2cm}
	\noindent
{\bf Proof of Theorem \ref{exidisspc_surfcant}:} 
	Consider the quadratic form  
	\begin{equation*}
		\hat{q}_\beta (\psi) = \hat{Q}_\beta(\psi) - E_1(\beta)\|\psi\|_{L^2({\cal S}^+)}^{2}, \quad \dom \hat{q}_\beta = \dom \hat{Q}_\beta.
	\end{equation*}
	According to (\ref{rayquoin}) and Theorem \ref{propress_surfcant}, it is enough to show that there exists a non null function $\psi \in \dom \hat{q}_\beta$ for which $\hat{q}_\beta(\psi) < 0$. A function with this property will be constructed below.
	
	The first step is to find a sequence $\{\psi_n\}_{n \in \mathbb{N}} \subset \dom q_\beta$ such that $q_\beta (\psi_n) \to 0$, as $n \to \infty$. For that,   let $w \in C^\infty(\mathbb{R})$ a real function where $w=1$ for $x \in [-1,1]$, and $w=0$ for $x \in \mathbb{R} \setminus [-2,2]$. Define, for each $n \in \mathbb{N}$,
	\begin{equation*}
		w_n(x) = w \left(\frac{x}{n}\right) \quad \mbox{and} \quad \psi_n(x,t) = w_n(x) \chi_1 (t),
	\end{equation*}  
	where $\chi_1$ denotes  a denotes the normalized eigenfunction corresponding to the eigenvalue $E_1(\beta)$ of the operator $T(\beta)$. In particular 
    \begin{equation*}
	\hat{q}_{\beta}(\psi_n) =  \hat{Q}_{\beta}(\psi_n) - E_1(\beta)\|\psi_n\|_{L^2(\Sigma)}^{2}
= \int_{\Sigma}  \Bigg[  \frac{|w'_n|^2 |\chi_1|^2}{h_\beta^2}  - \frac{2g'(x) \xi'_2(t)}{h_\beta^2} \left(w'_n \chi_1 w_n  \chi'_1  - \frac{ (h_\beta^2)'}{4 h_\beta^2} w_n w'_n |\chi_1|^2 \right) \Bigg] \dx \dt.
	\end{equation*}
Note that
\begin{equation*}
     \int_{\cal C}  \frac{ \xi'_2(t)}{h_\beta^2} \chi_1  \chi'_1\dt 
     = -\frac{1}{2}\int_{\cal C} \left(\frac{\xi'_2(t)}{h_\beta^2} \right)' |\chi_1|^2 \dt.
\end{equation*}
It can be verified that
 	\begin{equation*}
	\hat{q}_{\beta}(\psi_n) =  A\int_{\mathbb{R}}   |w'_n|^2 \dx + B\int_{\mathbb{R}} g'(x)  w_n w'_n \dx  \leq A \int_{\mathbb{R}}   |w'_n|^2 \dx + \beta^2 |B|\left(\int_{\mathbb{R}} |w_n|^2 \dx \right)^{1/2} \left(  \int_{\mathbb{R}} |w^{\prime}_n|^2 \dx \right)^{1/2}.
	\end{equation*}
 By (\ref{estdericuoff1_surf}) and since $\|w_n\|_\infty \leq 1$, for all $n \in \mathbb{N}$, by dominated convergence theorem, 
\begin{equation*}
    \hat{q}_\beta(\psi_n) \to 0, \quad \mbox{as} \quad n \to \infty. 
\end{equation*}

Now, fix $\varepsilon \in \mathbb{R}$.  For each $n \in \mathbb{N}$, define 
	\begin{equation*}
		\psi_{n,\varepsilon}(x,t) :=  \psi_n(x,t) + \varepsilon \phi(x,t).  
	\end{equation*}
	for some $\phi \in \dom \hat{q}_\beta$. In this case,
	\begin{equation*}
		\hat{q}_\beta(\psi_{n,\varepsilon}) = \hat{q}_\beta(\psi_n) + 2 \varepsilon \,  {\rm Re} \, (\hat{q}_\beta(\psi_n, \phi))
		+ \varepsilon^2 \hat{q}_\beta(\phi).
	\end{equation*}

The strategy is to show that there exists $\phi$ satisfying 
	\begin{equation}\label{rbetaneez_surfcant}
		\lim_{n \to \infty} \hat{q}_\beta(\psi_n, \phi) \neq 0.
	\end{equation}	
	In fact, if (\ref{rbetaneez_surfcant}) holds true, it is enough to chosen  $\varepsilon$ such that $\hat{q}_\beta(\psi_{n,\varepsilon})<0$, for some $n$ large enough. 
	
 	Consider $\eta \in C^\infty_0 (\mathbb{R})$, with $\operatorname{supp} \eta \subset (-1,1)$ and $\eta(0) \neq 0$.  Define $\phi(x,t) = \eta(x) \chi_1(t)$. One has 
\begin{align*}
    \hat{q}_{\beta}(\psi_n, \phi )  = &  \int_{\Sigma} \frac{1}{h_\beta^2 } \left( \psi_n' - g' \xi'_2 \left[\partial_t \psi_n - \frac{  (h_\beta^2)'}{4 h_\beta^2} \psi_n  \right] \right) \left( \phi' -  g' \xi'_2 \left[\partial_t \phi - \frac{ (h_\beta^2)'}{4 h_\beta^2} \phi  \right] \right)  \dx \dt\\
&+ \int_{\Sigma} 
  \left( \partial_t \psi_n - \frac{(h_\beta^2)'}{4 h^2} \psi_n  \right) \left( \partial_t \phi - \frac{(h_\beta^2)'}{4 h_\beta^2} \phi  \right) \dx \dt - E_1(\beta) \int_{\Sigma} \psi_n \phi \, \dx \dt\\
   = & \int_{\Sigma} \frac{1}{h_\beta^2 } \Bigg(  w'_n \eta' |\chi_1|^2  -  g' \xi'_2 w'_n \chi_1 \left[\eta \chi'_1 - \frac{(h_\beta^2)'}{4 h_\beta^2} \eta \chi_1  \right]  -  g' \xi'_2\left[ w_n  \chi'_1 \eta' \chi_1 - \frac{(h_\beta^2)'}{4 h_\beta^2} w_n  \eta' |\chi_1|^2 \right] \Bigg) \dx \dt\\
&+ \int_{\Sigma} \frac{(1 + \beta^2)}{h^2}
  \left( w_n \chi'_1 - \frac{(h_\beta^2)'}{4 h_\beta^2} w_n \chi_1  \right) \left( \eta \chi'_1 - \frac{(h_\beta^2)'}{4 h_\beta^2} \eta \chi_1   \right) \dx \dt - E_1(\beta) \int_{\Sigma} w_n  \eta |\chi_1|^2 \dx \dt\\
   \to &   \int_{\Sigma} \eta'  g'\left[ \left(\frac{ \xi'_2}{2h_\beta^2} \right)'  +   \frac{\xi'_2 (h_\beta^2)'}{4 h_\beta^4}   \right] |\chi_1|^2 \dx \dt + \int_{\Sigma} \frac{(1+ \beta^2)}{h^2} \eta
  \left( \chi'_1 - \frac{\partial_t h_\beta^2}{4 h_\beta^2} \chi_1  \right)^2 \dx \dt - E(\beta) \int_{\Sigma} \eta |\chi_1|^2 \dx \dt\\
  = & \frac{B}{2} \int_{\mathbb{R}} \eta' g'  \dx,
\end{align*}
as $n \to \infty$.  Finally,  since $\eta(0) \neq 0$, $B/2\neq 0$, 
\begin{equation*}
    \frac{B}{2} \int_{\mathbb{R}} \eta' g'  \dx = -\frac{\beta B}{2} \int_{-1}^0 \eta'  \dx + \frac{\beta B}{2} \int_{0}^1 \eta'  \dx = -\beta B \eta(0) \neq 0. 
\end{equation*}  
 Thus, \eqref{rbetaneez_surfcant} holds true. \qed

	%\section{Appendix}\label{appendix001}
	
	\vspace{0.2cm}
	\noindent
	{\bf Acknowledgments}

	\vspace{0.2cm}
	\noindent

The author is grateful to Alessandra Verri for useful suggestions, discussions, very detailed and critical reading of the manuscript. The work has been supported by CAPES (Brazil) through the process: 88887.511866/2020-00.

\vspace{0.2cm}
	\noindent 
    
	\bibliographystyle{plain}
%\bibliography{bibliodiana}

	\end{document}